\documentclass[]{article}

\usepackage{graphicx}
\usepackage{cancel}
\usepackage{authblk}
\usepackage{amsmath}
\usepackage{amsthm}
\usepackage{amssymb}
\newtheorem{theorem}{Theorem}
\newtheorem{corollary}{Corollary}
\usepackage[sort&compress,numbers]{natbib}
\usepackage[colorlinks=true,linkcolor=black, citecolor=blue, urlcolor=blue]{hyperref}

\begin{document}

\title{Complete general solutions for equilibrium equations of isotropic strain gradient elasticity} 

\author[1,2]{Y.O. Solyaev}

\affil[1]{Institute of Applied Mechanics of Russian Academy of Sciences, Moscow, Russia}
\affil[2]{Moscow Aviation Institute, Moscow, Russia}

\setcounter{Maxaffil}{0}
\renewcommand\Affilfont{\itshape\small}

\date{\today}

\maketitle

\begin{abstract}
In this paper, we consider isotropic Mindlin-Toupin strain gradient elasticity theory in which the equilibrium equations contain two additional length-scale parameters and have the fourth order. For this theory we developed an extended form of Boussinesq-Galerkin (BG) and Papkovich-Neuber (PN) general solutions. Obtained form of BG solution allows to define the displacement field through the single vector function that obeys the eight-order bi-harmonic/bi-Helmholtz equation. The developed PN form of the solution provides an additive decomposition of the displacement field into the classical and gradient parts that are defined through the standard Papkovich stress functions and modified Helmholtz decomposition, respectively. Relations between different stress functions and completeness theorem for the derived general solutions are established. As an example, it is shown that a previously known fundamental solution within the strain gradient elasticity can be derived by using the developed PN general solution.
\end{abstract}

\section{Introduction}
\label{int}

The development and the proof of completeness of general solutions in classical elasticity were the subject of the research during more than a hundred years. The most famous general solutions are known after Boussinesq and Galerkin \cite{boussinesq1885application, galerkin1930investigation} and Papkovich and Neuber \cite{papkovich1932representation, neuber1934neuer}. 
Interrelations between these solutions and their completeness have been discussed by Mindlin \cite{mindlin1936note}, Gurtin \cite{gurtin1962helmholtz}, Noll \cite{noll1957verschiebungsfunktionen}, Sternberg and Eubanks \cite{sternberg1955concept, eubanks1956completeness}, Sokolnikoff \cite{sokolnikoff1956mathematical}, Slobodyansky \cite{Slobodyansky1959}, Wang and Wang \cite{wang1991transformation}  and others. The universal constructive scheme for the development of general solutions and evaluation of its completeness and non-uniqueness within the classical   elasticity have been established based on the matrix methods of the theory of differential operators by Lurie, Wang and others \cite{lur1937theory, wang1995completeness, wang2008recent}.

In the present study, we consider the strain gradient elasticity theory (SGET), which general formulation for isotropic materials have been developed by Mindlin \cite{Mindlin1964} and Toupin \cite{toupin1964theories}. The main feature of SGET is the assumption that the potential-energy density of the media depends on the gradient of strain in addition to strain.  
In constitutive equations of SGET there arise five additional material constants in addition to two classical Lame parameters for isotropic materials, though only two additional parameters arise in the equilibrium equations \cite{dell2009generalized}. Boundary-value problem of SGET consists of the fourth-order equilibrium equations and extended number of boundary conditions, which form can be obtained based on the variational approach \cite{Mindlin1964, toupin1964theories}. In SGET there arise an extended definition of surface traction accounting for the normal and for the curvature of the Cauchy cut as well as additional definitions for double tractions and edge tractions \cite{Mindlin1964, toupin1964theories, dell2016cauchy}. The number of equilibrium equations in SGET remain the same to classical elasticity since the number of primary field variables (components of the displacement vector) does not change.

Notably, that Mindlin-Toupin SGET contains as the special cases a number of famous simplified gradient theories \cite{askes2011gradient, polizzotto2017hierarchy, gusev2017symmetry} and also several kinds of incomplete gradient theories like the couple stress theory, the dilatation gradient elasticity, etc. \cite{mindlin1962effects, dell2009generalized, lurie2021dilatation, eremeyev2020well, eremeyev2018linear}. Nowadays, these theories attracts an increasing attention in applications to fracture mechanics and dislocations problems \cite{gourgiotis2009plane, askes2015understanding, vasiliev2021new, lazar2013fundamentals, makvandi2021strain}, in the studies of small-scale and high-frequency processes \cite{cordero2016second, rosi2018validity}, in the description of mechanical behaviour of composites and metamaterails\cite{ma2014new, lurie2011eshelby, seppecher2019pantographic, dell2020discrete, solyaev2022self}. 

First variant of general solution for equilibrium equations of SGET have been presented in the initial work by Mindlin \cite{Mindlin1964}, though the particular variant of this solution within the couple stress theory have been established earlier by Mindlin and Tiersten \cite{mindlin1962effects}. The form of Mindlin solution \cite{Mindlin1964} can be treated as the generalized variant of the classical Papkovich-Neuber (PN) solution. It defines the displacement field in rather complicated form through the vector and scalar functions that obey the fourth order governing equations. Later, the simpler variants of the PN solutions with stress functions that obey the second-order equations have been established within SGET \cite{lurie2006interphase, lurie2011eshelby, charalambopoulos2015plane, charalambopoulos2022representing, solyaev2019three,placidi2017semi}, though the completeness of these solutions have not been proven or it was implied on the basis of heuristic reasoning  \cite{solyaev2022elastic}.

Lurie et al. \cite{lurie2006interphase, lurie2011eshelby} introduced PN general solution within the simplified variant of gradient theory and represented the displacement field through the additive decomposition into the sum of classical and gradient parts. The classical part of the solution was defined similarly to the standard PN solutions through the harmonic vector and scalar functions. Gradient part of the solution was defined by using two additional vector functions that satisfied the modified Helmholtz equations. 
Solyaev et al. \cite{solyaev2019three} used similar form of PN solution within SGET and reduced the representation for the gradient part of solution to the six scalar functions for the arbitrary curvilinear coordinates, which allows the separation of variables in the Helmholtz equation. 
Recently, this representation was additionally simplified such that the gradient part of PN solution was defined through the modified variant of Helmholtz decomposition, in which the scalar and the vector potentials obey the modified Helmholtz equations with different coefficients \cite{solyaev2022elastic}. Similar result have been also established within the simplified strain gradient elasticity theory by Charalambopoulos et al. \cite{charalambopoulos2022representing}, though the possibility of the use of solenoidal vector stress function for the gradient part of the displacement field have not been considered. Representation of the displacement field in terms of Lame's potentials have been used by Placidi and Dhaba within the Saint-Venant's problems \cite{placidi2017semi}.

In the present paper we derive the general solutions for equilibrium equations of Mindlin-Toupin SGET and prove its completeness in a deductive manner that was suggested by Papkovich \cite{papkovich1932representation} and used by Gurtin and Sternberg \cite{stern1962}, Gurtin \cite{gurtin1972} within the classical elasticity. We start from the Helmholtz decomposition for the displacement field, which is valid for the arbitrary smooth vector fields in the bounded domains and which validity in the infinite domains have been proven for weakly decaying fields in the classical work by Gurtin \cite{gurtin1962helmholtz}.
Based on Helmholtz decomposition and analysis of equilibrium equations of SGET we introduce the definition of the modified Galerkin vector that appears to obey the eight-order governing equation and that allows to define the displacement field in the presence of arbitrary bulk force, i.e. we obtaine the extended variant of Boussinesq-Galerkin (BG) solution within SGET. 
Generalizing the approach of classical elasticity, we introduce then the relations between the Galerkin vector and Papkovich-Neuber stress functions. 
The resulting modified form of PN solution becomes very attractive for applications because it provides a very simple form of the general solution of Mindlin-Toupin SGET with additive decomposition into the classical and gradient parts of the displacement field. The classical part is defined in a standard PN form through the harmonic vector and harmonic function, while the gradient part is simply represented through the modified Helmholtz decomposition. In such a way we prove the correctness of the previously supposed simplified form of PN solution within SGET \cite{solyaev2022elastic}. 

Combining both results for BG and PN forms of solution we prove the theorem of their completeness within SGET in a sense of similar proofs developed by Mindlin \cite{mindlin1936note} and Gurtin \cite{gurtin1972} within classical elasticity. Considered kind of a proof can be formulated according to Truesdell as follows: \textit{"corresponding to any stress field satisfying the given equations there exists at least one suitable choice of the stress functions"} \cite{truesdell1959invariant}. Hence, to prove the completeness of BG solution we show an explicit representation for the modified Galerkin  stress function through a given displacement field that satisfies the equilibrium equation of SGET. Similarly to classical elasticity, this is done by using the analogy between SGET equilibrium equation and relation between Galerkin stress function and the displacement field and also involving established PN representation of solution within SGET. Then, based on the relations between the stress functions of PN and BG solutions we prove completeness of the former. Thus, in the presented proof it is essential to have both representation of solutions in BG and PN forms.

The key point of the presented results is the proposed modified definition of Galerkin stress function and its specific relations to the Papkovich stress functions within SGET that have not been established previously for the best of author's knowledge. Applications of the obtained results can be related to the wide class of boundary value problems that can be solved analytically in a simple manner by using PN representation for which we show the completeness \cite{charalambopoulos2015plane, charalambopoulos2020plane, solyaev2022elastic}. Incorporation of complete general solutions into the numerical schemes (like in a Trefftz method) can be also an important issue for the development of stable and flexible numerical solvers within SGET \cite{lurie2006interphase,solyaev2021trefftz}.

\section{Preliminaries}
\label{pre}

Equilibrium equations of SGET can be represented in the following form \cite{Mindlin1964}:
\begin{equation}
	\alpha (1-l_1^2 \nabla^2) \nabla \nabla \cdot \pmb u
	- (1-l_2^2 \nabla^2) \nabla \times \nabla \times \pmb u = 
	-\frac{\pmb b}{\mu}
\label{ee}
\end{equation}
where  $\pmb u(\pmb{r})$ is the vector of mechanical displacements at a point $\pmb r = \{x_1, x_2, x_3\}$; $\pmb b(\pmb{r})$ is the body-force density vector; $\alpha = (\lambda+2\mu)/\mu$ = $2(1-\nu)/(1-2\nu)$ is classical non-dimensional parameter; $\lambda$, $\mu$ are the classical Lame constants; $\nu$ is the Poisson's ratio; $l_1$ and $l_2$ are the length scale parameters of isotropic elastic material that arise in the equilibrium equations of Mindlin-Toupin strain gradient elasticity.

Note, that the length scale parameters $l_1$ and $l_2$ are differently defined through the additional gradient material constants within the so-called Mindlin Forms I, II and III, that corresponds to the formulation of SGET in terms of second gradient of displacement, strain gradients and symmetric/anti-symmetric parts of strain gradients, respectively \cite{Mindlin1964}. Nevertheless, the form of equilibrium equations \eqref{ee} will be the same in all of these variants of SGET \cite{Mindlin1964}. Equilibrium equations of different simplified gradient theories can be obtained from Eq. \eqref{ee} as the particular cases. For example, in the Aifantis theory it is valid $l_1=l_2$ \cite{askes2011gradient}, in the couple stress and in the modified couple stress theory $l_1=0$ \cite{mindlin1962effects,munch2017modified}, and in the dilatation gradient elasticity $l_2=0$ \cite{lurie2021dilatation}. Assuming that both length scale parameters equal to zero $l_1=l_2=0$ equation \eqref{ee} reduces to classical elasticity equilibrium equation. Therefore, all considerations presented below for equation \eqref{ee} will be valid for any kind of mentioned gradient theories and all presented solutions must contain the corresponding classical solutions as the particular cases.

For the purpose of the following derivations we also need to define the general solution for the fourth-order scalar (or vector) equation of the following form:
\begin{equation}
	(1-l^2 \nabla^2) \nabla^2 \mathcal F  = \mathcal B
\label{eq}
\end{equation}
where $\mathcal F$ is the scalar (or the vector) field that should be found in the bounded or in the unbounded region of three-dimensional euclidean space $D$; $\mathcal B$ is the prescribed continuous scalar (or vector) field of class $C^n$ ($n\geq 1$) defined in $D$; and we assume that $l\in\mathbb R$ so that operator $(1-l^2 \nabla^2)$ corresponds to the modified Helmholtz equation also known as the screened Poisson equation. 

\begin{theorem}\label{thm1}
General solution of equation \eqref{eq} is given by:
\begin{equation}
	\mathcal F  = \mathcal F_c + \mathcal F_g
\label{t11}
\end{equation}
where $\mathcal F_c$ and $\mathcal F_g$ are the general solutions of Poisson equation and inhomogeneous modified Helmholtz equations, respectively
\begin{equation}
	\nabla^2\mathcal F_c =\mathcal B, \qquad (1-l^2 \nabla^2)\mathcal F_g = l^2\mathcal B
\label{t12}
\end{equation}
\end{theorem}

\begin{proof}
Let us define $\mathcal F_c$ as 
\begin{equation}
	\mathcal F_c = (1-l^2 \nabla^2)\mathcal F
\label{t13}
\end{equation}
Then from \eqref{eq}, \eqref{t13} it follows that $\mathcal F_c$ is the general solution of the Poisson's equation:
\begin{equation}
	\nabla^2\mathcal F_c =  \mathcal B
\label{t14}
\end{equation}
and we can define:
\begin{equation}
	\mathcal F_c = \bar{\mathcal F}_c + \mathcal F_c^*,
	\quad
	\nabla^2 \bar {\mathcal F}_c = 0,
	\quad
	\nabla^2 \mathcal F_c^* =  \mathcal B,
\label{t15}
\end{equation}
where $\bar{\mathcal F}_c$ is the general solution of Laplace equation and $\mathcal F^*_c$ is some particular solution of Eq. \eqref{t14}.

Then we represent the general solution of inhomogeneous Helmholtz equation \eqref{t13} in the following form:
\begin{equation}
	\mathcal F = \bar {\mathcal F} + \mathcal F^*, \quad
	(1-l^2 \nabla^2)\bar {\mathcal F} = 0, \quad
	(1-l^2 \nabla^2) \mathcal F^* = \mathcal F_c
\label{t16}
\end{equation}
where $\bar {\mathcal F}$ is the general solution of homogeneous Helmholtz equations and particular solution $\mathcal F^*$ can be defined as 
\begin{equation}
	\mathcal F^* = \mathcal F_c + \mathcal F_g^*,
\label{t167}
\end{equation}
where $\mathcal F_g^*$ is some unknown function. 

Substituting \eqref{t167} into \eqref{t16}$_3$ we find that $\mathcal F_g^*$ should be the particular solution of the following inhomogeneous Helmholtz equation:
\begin{equation}
	(1-l^2 \nabla^2) \mathcal F_g^* = l^2\mathcal B
\label{t17}
\end{equation}
where we take into account \eqref{t15}.

Then, introducing $\mathcal F_g = \bar {\mathcal F} + \mathcal F_g^*$ and using \eqref{t16}, \eqref{t17} we can identify $\mathcal F_g$ as the general solution of inhomogeneous Helmholtz equation 
\begin{equation}
	(1-l^2 \nabla^2) \mathcal F_g = l^2\mathcal B
\label{t18}
\end{equation}
and according to given definitions \eqref{t16}, \eqref{t167} we obtain $\mathcal F = \mathcal F_c + \mathcal F_g$.
\end{proof}

\begin{corollary}\label{cor1}
General solutions for equations \eqref{t12} can be presented in the following form:
\begin{equation}
F_c =(1-l^2 \nabla^2)\mathcal F, \qquad 
\mathcal F_g = l^2\nabla^2\mathcal F
\end{equation}
\end{corollary}

\begin{proof}
The proof follows from definitions \eqref{t11} and \eqref{t13}.
\end{proof}

Corollary \ref{cor1} can be treated as the completeness theorem for representation of general solution \eqref{t11} to equation \eqref{eq} since it shows how to define the parts of solution $\mathcal F_c$ and $\mathcal F_g$ for given arbitrary field $\mathcal F$ that satisfies \eqref{t11}.
For the further analysis we will also need the explicit form of particular solutions of equations \eqref{t12}, that can be represented as follows \cite{jeffreys1956methods}:
\begin{equation}
	\mathcal F_c^*(\pmb r) = (\mathcal N \ast \mathcal B) (\pmb r) = 
	-\frac{1}{4\pi}
	\int\limits_D 
	\frac{\mathcal B(\pmb \xi)}
				{\lvert\pmb r - \pmb \xi\rvert} 
			dv_{\pmb\xi}
\label{psp}
\end{equation}
\begin{equation}
	\mathcal F_g^*(\pmb r) = (\mathcal H \ast \mathcal B) (\pmb r)= 
		\frac{1}{4\pi}\int\limits_D 
			\frac{
			e^{-\lvert\pmb r - \pmb \xi\rvert/l}
			\,\mathcal B(\pmb \xi)}
				{\lvert\pmb r - \pmb \xi\rvert} dv_{\pmb\xi}
\label{psh}
\end{equation}
where $\mathcal N(\pmb r,\pmb\xi)$ is the Newtonian potential and $\mathcal H(\pmb r,\pmb\xi)$ is the Green's function of inhomogeneous modified Helmholtz equation \eqref{t12}$_2$.\\

\textit{Notations}. 

In the following derivations we will use bold symbols for vectors and italic symbols for scalars. General solutions of homogeneous  equations will be denoted with bar symbols ($\bar\psi$, $\bar{\pmb\Psi}$). Particular solutions will be denoted with star superscripts ($\psi^*$, $\pmb\Psi^*$). Potentials (stress functions) that satisfy the Laplace or the Poisson equations will be denoted with subscript "c" -- "classic" ($\psi_c$, $\pmb B_c$). Potentials that satisfy the modified Helmholtz equation will be denoted with subscript "g" -- "gradient" ($\psi_g$, $\pmb B_g$).

\section{Boussinesq-Galerkin solution}
\label{bsg}

According to Helmholtz theorem every suitably regular vector field $\pmb u(\pmb r)$ admits the representation:
\begin{equation}
	\pmb u = \nabla \phi + \nabla\times\pmb S
\label{ht}
\end{equation}
where $\phi$ is the scalar potential and $\pmb S$ is the vector potential for which we can assume that $\nabla\cdot\pmb S=0$ without loss of generality.

For the further analysis it will be enough to assume that in the infinite domains representation \eqref{ht} can be introduced under assumptions of weak decay conditions for the displacement field $\pmb u$ that are used in classical elasticity \cite{gurtin1962helmholtz}, though the validity of Helmholtz theorem for the fields with sub-linear growth have been also established \cite{petrascheck2015helmholtz,petrascheck2017helmholtz}.

Substituting \eqref{ht} into \eqref{ee} and using standard vector calculus identities we obtain equilibrium equations in terms of potentials:
\begin{equation}
	\nabla^2 \left(
	\alpha (1-l_1^2 \nabla^2) \nabla\phi
	+ (1-l_2^2 \nabla^2) \nabla \times \pmb S
	\right) = -\frac{\pmb b}{\mu}
\label{eefs}
\end{equation}

Equation \eqref{eefs} can be reduced to the high-order equation with respect to the single vector function $\pmb W(\pmb r)$ by using the following definitions of potentials $\phi$ and $\pmb S$:
\begin{equation}
\begin{aligned}
	\phi &= \frac{1}{\alpha}(1-l_2^2 \nabla^2)\nabla\cdot \pmb W\\
	\pmb S &= -(1-l_1^2 \nabla^2)\nabla\times \pmb W
\label{fsw}
\end{aligned}
\end{equation}

Combining \eqref{eefs}, \eqref{fsw} we obtain:
\begin{equation}
\begin{aligned}
	(1-l_1^2 \nabla^2)(1-l_2^2 \nabla^2)\nabla^2 \nabla^2 \pmb W 
		= -\frac{\pmb b}{\mu}
\label{eebg}
\end{aligned}
\end{equation}

Using \eqref{ht}, \eqref{fsw} we found:
\begin{equation}
	\pmb u = \frac{1}{\alpha}(1-l_2^2 \nabla^2)\nabla\nabla\cdot \pmb W
			-(1-l_1^2 \nabla^2)\nabla\times\nabla\times \pmb W
\label{bg}
\end{equation}
or alternatively:
\begin{equation}
	\pmb u = (1-l_1^2 \nabla^2)\nabla^2\pmb W - \kappa(1-l_3^2 \nabla^2)\nabla\nabla\cdot \pmb W
\label{bg1}
\end{equation}
where $\kappa = \frac{\alpha-1}{\alpha} = \frac{1}{2(1-\nu)}$ and $l^2_3 = \frac{\alpha}{1-\alpha}\left(\frac{1}{\alpha} l_2^2 - l_1^2\right)$ (it can be shown that $l_3\in\mathbb R$ that is the consequence of definitions of $l_1$, $l_2$ and requirements for the positive definition of strain energy density in isotropic SGET, see \cite{dell2009generalized}).

Representation \eqref{bg1} (or \eqref{bg}) should be treated as Boussinesq-Galerkin (BG) solution generalized for SGET. Vector fuction $\pmb W$ is Galerkin stress function that obey the eight-order bi-harmonic/bi-Helmholtz  governing equation \eqref{eebg}. In absence of gradient effects ($l_1=l_2=0$) this representation reduces to classical BG solution. The key idea in the presented form of BG solution is the appropriate choice of definitions for the Helmholtz potentials \eqref{fsw}.

Notably, that bi-Helmholtz/bi-Laplace equation have been considered previously within the second strain gradient elasticity theory (accounting for the dependence of strain energy on strain, gradient of strain and second gradient of strain) where it was shown that similar eight-order equation defines the fundamental solution for the modified Airy stress function within the edge dislocations problems \cite{lazar2019non, lazar2006dislocations}. 

\section{Papkovich-Neuber solution}
\label{bpn}

Let us introduce the following vector function:
\begin{equation}
	\pmb B = (1-l_1^2 \nabla^2)\nabla^2\pmb W
\label{bw}
\end{equation}

From equilibrium equations written in terms of Galerkin stress function \eqref{eebg}, we found that $\pmb B$ has to satisfy
\begin{equation}
	(1-l_2^2 \nabla^2)\nabla^2\pmb B = -\frac{\pmb b}{\mu}
\label{eeb}
\end{equation}

By making use of Theorem \ref{thm1}, general solution of equation \eqref{eeb} can be constructed as:
\begin{equation}
	\pmb B = \pmb B_c + \pmb B_g
\label{bcg}
\end{equation}
where
\begin{equation}
	\nabla^2\pmb B_c = -\frac{\pmb b}{\mu}, \qquad
	(1-l_2^2\nabla^2)\pmb B_g = -l_2^2\frac{\pmb b}{\mu}
\label{bth}
\end{equation}

and according to Corollary \ref{cor1} we also have
\begin{equation}
	\pmb B_c = (1-l_2^2\nabla^2)\pmb B, \qquad
	\pmb B_g = l_2^2\nabla^2\pmb B
\label{bcbg}
\end{equation}

Then, let us define the general solution of equation \eqref{bw}:
\begin{equation}
	\pmb W = \bar {\pmb W} + \pmb W^*
\label{wgen}
\end{equation}
where $\bar {\pmb W}$ is the general solution of corresponding homogeneous equations and $\pmb W^*$ is appropriate particular solution, i.e.:
\begin{equation}
	(1-l_1^2 \nabla^2)\nabla^2 \bar {\pmb W} = 0, \qquad
	(1-l_1^2 \nabla^2)\nabla^2 \pmb W^* = \pmb B
\label{wgwp}
\end{equation}

Substituting \eqref{wgen} into BG solution for the displacement field \eqref{bg1} and taking into account \eqref{wgwp} we obtain:
\begin{equation}
\label{uww}
\begin{aligned}
	\pmb u &= (1-l_1^2 \nabla^2)(\nabla^2\bar {\pmb W} + \nabla^2\pmb W^*) - \kappa(1-l_3^2 \nabla^2)\nabla
				(\nabla\cdot \bar {\pmb W} + \nabla\cdot \pmb W^*)\\
			&= \pmb B - \kappa\nabla(\bar \varphi + \varphi^*)
\end{aligned}
\end{equation}
where we introduce new scalar potentials:
\begin{equation}
\label{pp}
\begin{aligned}
	\bar \varphi = (1-l_3^2 \nabla^2)(\nabla\cdot \bar {\pmb W}), \qquad
	\varphi^* = (1-l_3^2 \nabla^2)(\nabla\cdot \pmb W^*)
\end{aligned}
\end{equation}

Using \eqref{wgwp}$_1$, \eqref{pp}$_1$ we immediately find that potential $\bar \varphi$ obeys the following homogeneous equation:
\begin{equation}
\label{phi}
\begin{aligned}
	(1-l_1^2 \nabla^2)\nabla^2\bar \varphi = 0
\end{aligned}
\end{equation}
which general solution can be found based on Theorem \ref{thm1}.

In order to derive the representation for potential $\varphi^*$ let us consider the relation, that follows from \eqref{wgwp}$_2$, \eqref{pp}$_2$:
\begin{equation}
\label{psi}
\begin{aligned}
	(1-l_1^2 \nabla^2)\nabla^2\varphi^*
	&= (1-l_1^2 \nabla^2)(1-l_3^2 \nabla^2)\nabla^2(\nabla\cdot \pmb W^*)\\
	&= (1-l_3^2 \nabla^2)(\nabla\cdot \pmb B)
\end{aligned}
\end{equation}

From this relation it is seen that potential $\varphi^*$ should be treated as the particular solution of equation \eqref{psi} since the general solution of corresponding homogeneous equation is given by $\bar\varphi$ \eqref{phi} and it is already included into the displacement representation \eqref{uww} (similar particular solution in classical PN solution obeys Poisson equation). Substituting relation between $l_3$ and $l_1$, $l_2$ into \eqref{psi} we can find then
\begin{equation}
\label{psi1}
\begin{aligned}
	(1-l_1^2 \nabla^2)\nabla^2\varphi^* 
	= 
	\frac{1}{\kappa}(1-l_1^2 \nabla^2)(\nabla\cdot \pmb B)
	-\frac{1}{\kappa\alpha}(1-l_2^2 \nabla^2)(\nabla\cdot \pmb B)
\end{aligned}
\end{equation}

Using representation \eqref{bcg}, \eqref{bth} for the second term in the right hand side of equation \eqref{psi1} we obtain:
\begin{equation}
\label{psi2}
\begin{aligned}
	(1-l_1^2 \nabla^2)\nabla^2\varphi^*
	&= 
	\frac{1}{\kappa}(1-l_1^2 \nabla^2)(\nabla\cdot \pmb B)
	-\frac{1}{\kappa\alpha}
		\nabla\cdot \left(\pmb B_c + l_2^2 \frac{\pmb b}{\mu} \right)
	+\frac{l_2^2}{\kappa\alpha} \frac{\nabla\cdot\pmb b}{\mu}\\
	\implies
	(1-l_1^2 \nabla^2)\nabla^2\varphi^*
	&= 
	\frac{1}{\kappa}(1-l_1^2 \nabla^2)(\nabla\cdot \pmb B)
	-\frac{1}{\kappa\alpha}\nabla\cdot \pmb B_c\\
\end{aligned}
\end{equation}
 
Then, we take into account that $\pmb  B_c$ is the general solution for Poisson equation \eqref{bth}$_2$. Therefore, without loss of generality vector $\pmb B_c$ can be replaced by 
$$(1 -l_1^2 \nabla^2) \pmb B_c - l_1^2 \frac{\pmb b}{\mu}$$

In such a way from \eqref{psi2} we obtain the final form of the governing equation for potential $\varphi^*$:

\begin{equation}
\label{psi3}
\begin{aligned}
	(1-l_1^2 \nabla^2)\nabla^2\varphi^* 
	&= 
	\frac{l_1^2}{\kappa\alpha\mu} \nabla\cdot\pmb b
	+
	(1-l_1^2 \nabla^2)
	\nabla\cdot
	\left(
		\pmb B_c 
		+\frac{1}{\kappa}\pmb B_g
	\right) 
\end{aligned}
\end{equation}
where we also take into account decomposition \eqref{bcg} and definitions for $\alpha$ and $\kappa$.

Particular solution for the obtained equation \eqref{psi3} can be decomposed into three parts:
\begin{equation}
\label{psi4}
\begin{aligned}
	\varphi^* = \varphi^*_b + \varphi_c^* + \varphi_g^*
\end{aligned}
\end{equation}
so that
\begin{equation}
\label{psi8}
	(1-l_1^2 \nabla^2)\nabla^2\varphi^*_b 
	= \frac{l_1^2}{\kappa\alpha\mu} \nabla\cdot\pmb b
\end{equation}
\begin{equation}
\label{psi6}
	(1-l_1^2 \nabla^2)\nabla^2\varphi_c^* = (1-l_1^2 \nabla^2)
	\nabla\cdot\pmb B_c 
\end{equation}
\begin{equation}
\label{psi7}
	(1-l_1^2 \nabla^2)\nabla^2\varphi_g^* = \frac{1}{\kappa}(1-l_1^2 \nabla^2)
	\nabla\cdot\pmb B_g
\end{equation}

Definition of particular solution $\varphi_c^*$ \eqref{psi6} can be reduced to the classical problem, in which similar particular solution was found for equation $\nabla^2\varphi_c^* = \nabla\cdot\pmb B_c $ in the following form \cite{papkovich1932representation, gurtin1972}: 
\begin{equation}
\label{psi9}
	\varphi_c^* = \frac{1}{2}(\pmb r\cdot \pmb B_c + \beta_c),
\end{equation}
where
\begin{equation}
\label{psi90}
	\nabla^2\pmb B_c = -\frac{\pmb b}{\mu}, \qquad
	\nabla^2\beta_c = \frac{\pmb r\cdot\pmb b}{\mu}
\end{equation}

Particular solution $\varphi_g^*$ \eqref{psi7} can be defined by using Helmhlotz decomposition for the vector potential $\pmb B_g$, which can be used without any additional restrictions since $\pmb B_g$ satisfies the vector Helmholtz equation \eqref{bth}$_2$ \cite{morse1954methods}. Thus, we define:
\begin{equation}
\label{psi10}
\begin{aligned}
	\pmb B_g = \nabla \Psi + \nabla\times \pmb\Psi,\qquad
	\nabla\cdot\pmb B_g = \nabla^2 \Psi, \qquad
	\nabla\times\pmb B_g = -\nabla^2 \pmb\Psi
\end{aligned}
\end{equation}
where function $\Psi$ and solenoidal field $\pmb\Psi$ are the scalar and vector potentials, respectively, that have to satisfy the following equations:
\begin{equation}
\label{psi11}
\begin{aligned}
	(1-l_2^2\nabla^2)\nabla^2\Psi = -l_2^2\frac{\nabla\cdot\pmb b}{\mu}, \qquad
	(1-l_2^2\nabla^2)\nabla^2\pmb\Psi = l_2^2\frac{\nabla\times\pmb b}{\mu},
\end{aligned}
\end{equation}
which are obtained based on equation \eqref{bth}$_2$ and definitions \eqref{psi10}.

Based on comparison of equation for the scalar potential $\Psi$ \eqref{psi10}$_{2}$ and equation for the particular solution $\varphi_g^*$ \eqref{psi7} we can define the last one as:

\begin{equation}
\label{psi12}
	\varphi_g^* = \frac{1}{\kappa}\Psi
\end{equation}

Then, substituting \eqref{bcg}, \eqref{psi4} into \eqref{uww} and taking into account \eqref{psi9}, \eqref{psi12} we obtain:
\begin{equation}
\label{upn0}
	\pmb u = \pmb B_c + \pmb B_g 
	- \kappa\nabla\left(
		\bar\varphi 
		+ \varphi^*_b
		+ \frac{1}{2}(\pmb r\cdot \pmb B_c + \beta_c)
		+ \frac{1}{\kappa}\Psi
	\right)
\end{equation}

In this relation we observe that the sum $\bar\varphi + \varphi^*_b$ is the general solution of inhomogeneous equation \eqref{psi8}, that will be denoted in the following as $\varphi = \bar\varphi + \varphi^*_b$. Substituting Helmholtz decomposition for $\pmb B_g$ \eqref{psi10} into \eqref{upn0} we also find that the potential part of this field $\nabla\Psi$ is cancelled and definition of the displacement field becomes to:
\begin{equation}
\label{upn1}
	\pmb u = \pmb B_c 
	- \frac{\kappa}{2}\nabla(\pmb r\cdot \pmb B_c + \beta_c)
	+ \nabla\times \pmb\Psi
	- \kappa\nabla \varphi
\end{equation}
in which $\pmb B_c$ and $\beta_c$ should be treated as standard harmonic scalar and vector stress functions of PN solution defined by \eqref{psi90}; $\pmb\Psi$ and $\varphi$ are the additional stress functions of gradient theory; $\pmb\Psi$ is the general solutions of equation \eqref{psi11}$_2$ and $\varphi$ is the general solution of equation:

\begin{equation}
\label{scp}
	(1-l_1^2 \nabla^2)\nabla^2\varphi 
	= \frac{l_1^2}{\kappa\alpha\mu} \nabla\cdot\pmb b, 
\end{equation}

Obtained representation \eqref{upn1} can be additionally simplified. Namely, according to Theorem \ref{thm1} we can define:
\begin{equation}
\label{scp1}
\begin{aligned}
	\varphi &= \tilde\varphi_c + \varphi_g, \qquad
	\nabla^2\tilde\varphi_c = \frac{l_1^2}{\kappa\alpha\mu} \nabla\cdot\pmb b,
	\qquad
	(1-l_1^2\nabla^2)\varphi_g = 
		\frac{l_1^4}{\kappa\alpha\mu} \nabla\cdot\pmb b 
\end{aligned}
\end{equation}
so that the classical harmonic part of this solution $\tilde\varphi_c$ can be combined with corresponding standard Papkovich stress function $\beta_c$ in \eqref{upn1}. It is convenient to introduce then the following single scalar function in relation \eqref{upn1}:
\begin{equation}
\label{phic}
\varphi_c = \beta_c + 2\tilde\varphi_c
\end{equation}
 that should obey the Poisson equation with the right hand side defining by corresponding linear combination of the right hand sides of equations \eqref{scp1}$_2$ and \eqref{psi90}$_2$. 

In representation \eqref{upn1} we can also replace term $\nabla\times\pmb\Psi$ by the initial vector potential $\pmb B_g$ (see \eqref{psi10}) with additional requirement that this vector field should be solenoidal. As the result, after some simplifications and renormalization for the potential $\varphi_g$, one can obtain the following final form of Papkovich-Neuber general solution within SGET:
\begin{equation}
\label{pn}
\begin{aligned}
	\pmb u &= \pmb u_c + \pmb u_g,\\
	\pmb u_c &= \pmb B_c 
	- \frac{1}{4(1-\nu)}\nabla(\pmb r\cdot \pmb B_c + \varphi_c)\\
	\pmb u_g &= \pmb B_g + l_1^2\nabla \varphi_g
\end{aligned}
\end{equation}
in which the stress functions $\pmb B_c$, $\pmb B_g$, $\varphi_c$, $\varphi_g$ have to satisfy:
\begin{equation}
\label{pnpot}
\begin{aligned}
	\nabla^2\pmb B_c &= -\frac{\pmb b}{\mu}, \\
	\nabla^2\varphi_c &= 
			\frac{\pmb r\cdot\pmb b}{\mu} 
			+ 2l_1^2(1-2\nu)\frac{\nabla\cdot\pmb b}{\mu} \\
	(1-l_2^2\nabla^2)\pmb B_g &= -l_2^2\frac{\pmb b}{\mu}, \qquad
	\nabla\cdot\pmb B_g = 0, \\
	(1-l_1^2\nabla^2)\varphi_g &= -l_1^2\,\frac{1-2\nu}{2(1-\nu)}\, 
					\frac{\nabla\cdot\pmb b}{\mu}
\end{aligned}
\end{equation}

Obtained representation \eqref{pn} validates the possibility of additive decomposition of general solution for the displacement field within Mindlin-Toupin SGET into the so-called classical $\pmb u_c$ and gradient $\pmb u_g$ parts. Such representation have been heuristically assumed in Refs. \cite{solyaev2019three, solyaev2022elastic} and  and it was explicitly established previously only within the simplified gradient theories \cite{askes2011gradient,lurie2006interphase, charalambopoulos2022representing}. Similar decomposition have been also obtained within the analysis of fundamental solutions of different gradient theories \cite{gao2009green, gourgiotis2018concentrated, ma2018inclusion}. 

In the derived form of PN solution \eqref{pn} the classical part of the displacement field $\pmb u_c$ is defined through the standard Papkovich stress functions (with the only modification of the body force in the Poisson equation for the scalar function). The gradient part of the displacement field $\pmb u_g$ is represented via the modified Helmholtz decomposition that is the linear combination of the potential part $\nabla \varphi_g$ and solenoidal part $\pmb B_g$ defined as the solutions of modified Helmoltz equations with different length scale parameters (see \eqref{pnpot}). In this modified decomposition the length scale parameter $l_1$ defines the potential part of $\pmb u_g$, while $l_2$ defines its rotational part that is in agreement with the initial structure of equilibrium equations of SGET \eqref{ee}. 

General solutions for different simplified gradient theories can be obtained assuming corresponding values for the length scale parameters. Namely, assuming $l_1=0$ we obtain the general solution for the couple stress theory \cite{Mindlin1964}. For the case $l_2=0$ we obtain the general solution for the dilatation gradient elasticity theory \cite{lurie2021dilatation}. For the simplified Aifantis theory it should be used $l_1=l_2$. Classical PN solution follows from \eqref{pn}, \eqref{pnpot} if $l_1=l_2=0$.

\section{Completeness theorem}
\label{comp}

The main next step is to prove the completeness of the developed solutions. Reformulating the similar classical theorem given by Gurtin \cite{gurtin1972} we state:

\begin{theorem}\label{thm2}
Let $\pmb u$ be a displacement field that satisfies equilibrium equation \eqref{ee} and corresponds to the body force $\pmb b$. Then there exists a field $\pmb W$ that satisfies \eqref{eebg}, \eqref{bg1}; and fields $\pmb B_c$, $\pmb B_g$, $\varphi_c$, $\varphi_g$ that satisfy \eqref{pn}, \eqref{pnpot}.
\end{theorem}

\begin{proof} Consider the definition for the displacement field through the modified Galerkin stress function within SGET \eqref{bg}:
			
\begin{equation}
\label{t21}
	\pmb u = \frac{1}{\alpha}(1-l_2^2 \nabla^2)\nabla\nabla\cdot \pmb W
			-(1-l_1^2 \nabla^2)\nabla\times\nabla\times \pmb W
\end{equation}

This relation can be rewritten in the following equivalent form:

\begin{equation}
\label{t22}
\begin{aligned}
	\hat\alpha(1-\hat l_1^{\,2} \nabla^2)\nabla\nabla\cdot \pmb W
			-(1-\hat l_2^{\,2} \nabla^2)\nabla\times\nabla\times \pmb W 
			= -\frac{\hat {\pmb b}}{\mu}
\end{aligned}
\end{equation}
where $\hat\alpha = 1/\alpha$, $\hat{\pmb b} = -\mu\pmb u$, $\hat l_1 = l_2$, $\hat l_2 = l_1$.

The form of equation \eqref{t22} with respect to $\pmb W$ is exactly the same to the equilibrium equation of SGET that is defined with respect to $\pmb u$ \eqref{ee}. Therefore, the task of finding a field $\pmb W$ that satisfies \eqref{t21} is reduced to finding a particular solution to the equilibrium equation of SGET corresponding to given body forces $\hat {\pmb b}$. This can be done by using the derived form of Papkovich-Neuber solution of SGET \eqref{pn}, \eqref{pnpot}:
\begin{equation}
\label{t23}
\begin{aligned}
	\pmb W = \hat {\pmb B}_c 
	- \frac{1}{4(1-\hat\nu)}\nabla(\pmb r\cdot \hat {\pmb B}_c + \hat\varphi_c) 
	+ \hat {\pmb B}_g + \hat l_1^{\,2}\nabla \hat \varphi_g
\end{aligned}
\end{equation}
in which the stress functions can be defined based on \eqref{pnpot} and \eqref{psp}, \eqref{psh}:
\begin{equation}
\begin{aligned}
	\hat {\pmb B}_c &= \mathcal N \ast \pmb u, \qquad
	\hat\varphi_c = -\mathcal N \ast (\pmb r\cdot \pmb u 
			+ 2\hat l_1^{\,2}(1-2\hat\nu) \nabla\cdot\pmb u), \\
	\hat {\pmb B}_g &= \nabla\times 
					(\hat{\pmb\Psi}_c + \hat{\pmb\Psi}_g),
	\quad 
	\hat{\pmb\Psi}_c = -\hat l^{\,2}_2\, \mathcal N \ast (\nabla\times\pmb u),
	\quad
	\hat{\pmb\Psi}_g = -\hat l^{\,2}_2\, \mathcal H_2 \ast (\nabla\times\pmb u) \\
	\hat\varphi_g &= \hat l_1^{\,2}\,\frac{1-2\hat\nu}{2(1-\hat\nu)}
			\mathcal H_1 \ast \left(\nabla\cdot\pmb u \right)
\label{t24}
\end{aligned}
\end{equation}
where $\mathcal H_1$ and $\mathcal H_2$ are the Green's functions of Helmholtz equations \eqref{psh} defined with the length scale parameters $l_1$ and $l_2$, respectively; and for the solenoidal gradient potential we use its representation $\hat{\pmb B}_g =\nabla\times\hat{\pmb\Psi}$ in which the vector field $\hat{\pmb\Psi}$ is defined based on equation \eqref{psi11}$_2$ and Theorem \ref{thm1}.

Thus, we prove that there exist a vector field $\pmb W$ that corresponds to a given displacement field $\pmb u$. The governing equation for $\pmb W$\eqref{eebg} follows from equilibrium equation \eqref{ee} by its definition and consequently the completeness of BG solution \eqref{eebg}, \eqref{bg1} within SGET is proven. 

Consider then the relations between the Papkovich-Neuber and the Galerkin stress functions. By using relations \eqref{bw}, \eqref{bcbg} and taking into account that $\nabla\cdot\pmb B_g=0$, we  obtain the definitions for the vector stress functions of PN solution \eqref{pn}:

\begin{equation}
\begin{aligned}
	\pmb B_c &= (1-l_1^2 \nabla^2)(1-l_2^2\nabla^2)\nabla^2\pmb W, \\
	\pmb B_g &= -l_2^2\, (1-l_1^2 \nabla^2)\nabla^2(\nabla\times\nabla\times\pmb W)\\
\label{tbgbc}
\end{aligned}
\end{equation}

Scalar potentials $\varphi_c$ and $\varphi_g$ in PN solution \eqref{pn} can be defined based on relations \eqref{scp}-\eqref{phic} and Corollary \ref{cor1} as follows:

\begin{equation}
\begin{aligned}
	\varphi_c &= \beta_c + 2\tilde\varphi_c \\
	\varphi_g &= l_1^2\nabla^2\varphi\\
\label{fcfg1}
\end{aligned}
\end{equation}
where $\beta_c$ is indeterminate harmonic function and $\tilde\varphi_c$ has to satisfy
\begin{equation}
\begin{aligned}
	\tilde\varphi_c = (1-l_1^2\nabla^2)\varphi
\label{fcfg2}
\end{aligned}
\end{equation}

Scalar potential $\varphi$ was defined in Section 4 based on the following relation:
\begin{equation}
\begin{aligned}
	\varphi = \bar\varphi + \varphi^*_b
\label{fcfg3}
\end{aligned}
\end{equation}
in which according to \eqref{wgwp}$_1$:
\begin{equation}
\begin{aligned}
	\bar \varphi = (1-l_3^2 \nabla^2)(\nabla\cdot \bar {\pmb W})
\label{fcfg4}
\end{aligned}
\end{equation}
and according to \eqref{psi4}, \eqref{psi9}, \eqref{psi12}:
\begin{equation}
\begin{aligned}
	\varphi^*_b  = \varphi^* - \psi_c^* - \psi_g^*
= (1-l_3^2 \nabla^2)(\nabla\cdot \pmb W^*) - \frac{1}{2}(\pmb r\cdot \pmb B_c + \beta_c) - \frac{1}{\kappa}\Psi
\label{fcfg5}
\end{aligned}
\end{equation}

By its definition, $\pmb B_g$ \eqref{pnpot} does not have the potential part so that in \eqref{fcfg5} we can set $\Psi=0$ (see \eqref{psi10}). Then, substituting \eqref{fcfg4}, \eqref{fcfg5} into \eqref{fcfg3} and taking into account \eqref{wgen} we obtain:
\begin{equation}
\begin{aligned}
	\varphi = (1-l_3^2 \nabla^2)(\nabla\cdot \pmb W) 
			- \frac{1}{2}(\pmb r\cdot \pmb B_c + \beta_c)
\label{fcfg6}
\end{aligned}
\end{equation}

Using definition for $\tilde\varphi_c$ \eqref{fcfg2}, obtained relation  \eqref{fcfg6} and definition of particular solution $\varphi_c^*$ \eqref{psi6}, \eqref{psi9} we find that:
\begin{equation}
\begin{aligned}
	\tilde\varphi_c = 
		(1-l_1^2\nabla^2)(1-l_3^2 \nabla^2)(\nabla\cdot \pmb W) 
		- \frac{1}{2}(\pmb r\cdot \pmb B_c + \beta_c) 
		+ l_1^2\nabla\cdot\pmb B
\label{fcfg7}
\end{aligned}
\end{equation}

Finally, classical scalar PN stress functions $\varphi_c$ can be defined by using \eqref{fcfg1}$_1$, \eqref{fcfg7} as follows:
\begin{equation}
\begin{aligned}
	\varphi_c = 
2(1-l_1^2\nabla^2)(1-l_3^2 \nabla^2)(\nabla\cdot \pmb W) 
- \pmb r\cdot \pmb B_c 
+ 2l_1^2\nabla\cdot\pmb B_c
\label{fcfg8}
\end{aligned}
\end{equation}

Gradient scalar PN stress function $\varphi_g$ can be defined based on relations \eqref{fcfg1}$_2$, \eqref{fcfg6}:
\begin{equation}
\begin{aligned}
		\varphi_g = 
		l_1^2\nabla^2(1-l_3^2 \nabla^2)(\nabla\cdot \pmb W) 
		- l_1^2\nabla\cdot\pmb B_c
\label{fcfg9}
\end{aligned}
\end{equation}

Derived relations \eqref{tbgbc}, \eqref{fcfg8}, \eqref{fcfg9} allow us to define PN stress functions $\pmb B_c$, $\pmb B_g$, $\varphi_c$, $\varphi_g$ by using given Galerkin stress vector $\pmb W$, which representation for a given arbitrary displacement field $\pmb u$ have been already found \eqref{t24}. Then, according to derivations presented in previous section, equation \eqref{eebg} implies \eqref{pnpot}, and \eqref{bg1} implies \eqref{pn}. Therefore, PN solution \eqref{pn}, \eqref{pnpot} is also complete within SGET.
\end{proof}

\section{Fundamental solution}
\label{fund}
Derivation of the fundamental solution based on PN general solution within the classical elasticity can be found, e.g. in \cite{ 
eubanks1956completeness, gurtin1972}. Within SGET similar result can be obtained by using developed form of PN solution \eqref{pn}, \eqref{pnpot}. Thus, let us consider the concentrated body force $\pmb Q$ applied at the origin of coordinate system:
\begin{equation}
\label{bdel}
	\pmb b = \pmb Q \delta(\pmb r)
\end{equation}
where $\delta(\pmb r)$ is Dirac delta function.

Based on \eqref{pn} the fundamental (Kelvin) solution of SGET that corresponds to the body force \eqref{bdel} can be presented as follows:
\begin{equation}
\label{pnf}
\begin{aligned}
	\pmb u = \pmb B^*_c 
	- \frac{1}{4(1-\nu)}\nabla(\pmb r\cdot \pmb B^*_c + \varphi^*_c)
	+ \pmb B^*_g + l_1^2\nabla \varphi^*_g
\end{aligned}
\end{equation}
where the stress functions $\pmb B^*_c$, $\pmb B^*_g$, $\varphi^*_c$, $\varphi^*_g$ are the particular solutions of the following equations
\begin{equation}
\label{pnf1}
	\nabla^2\pmb B^*_c = -\frac{\pmb Q}{\mu}\delta(\pmb r),
\end{equation}
\begin{equation}
\label{pnf2}
	\nabla^2\varphi^*_c = \frac{\pmb r\cdot\pmb Q}{\mu}\delta(\pmb r)
						+ 2l_1^2(1-2\nu)
							\frac{\nabla\cdot(\pmb Q \delta(\pmb r))}{\mu},
\end{equation}
\begin{equation}
\label{pnf3}
	(1-l_2^2\nabla^2)\pmb B^*_g = -l_2^2\frac{\pmb Q}{\mu}\delta(\pmb r), \qquad
	\nabla\cdot\pmb B^*_g = 0, 
\end{equation}
\begin{equation}
\label{pnf4}
	(1-l_1^2\nabla^2)\varphi^*_g = - l_1^2\,\frac{1-2\nu}{2(1-\nu)}\, 
					\frac{\nabla\cdot(\pmb Q \delta(\pmb r))}{\mu}
\end{equation}

The remaining task is to found the particular solutions of equations \eqref{pnf1}-\eqref{pnf4}. Solution for vector stress function $\pmb B^*_c$\eqref{pnf1} is similar to classical elasticity and according to potential theory it is given by \cite{gurtin1972}:
\begin{equation}
\label{bcsol}
	\pmb B^*_c = \frac{\pmb Q}{4\pi\mu r},
\end{equation}
where $r = \lvert\pmb r\rvert$.

Scalar stress function $\varphi^*_c$ in classical elasticity vanishes since it is valid that:
\begin{equation}
\label{phicsol0}
	\nabla^2\varphi^*_c = \frac{\pmb r\cdot\pmb Q}{\mu}\delta(\pmb r)
	\quad \implies \quad \varphi^*_c\equiv 0
\end{equation}
therefore, equation \eqref{pnf2} can be defined in a simpler form as 
\begin{equation}
\label{pnf22}
	\nabla^2\varphi^*_c = 2l_1^2(1-2\nu)
							\frac{\nabla\cdot(\pmb Q \delta(\pmb r))}{\mu},
\end{equation}
Particular solution to this equation can be presented in the following form:
\begin{equation}
\label{phicsol}
	\varphi^*_c = -l_1^2\frac{1-2\nu}{2\pi\mu}\,\pmb Q\cdot\nabla
					\left(\frac{1}{r}\right)
\end{equation}
which follows from formula for the particular solutions of Poisson equation \eqref{psh} and the identity that can be proven based on divergence theorem:
\begin{equation}
\label{iden}
	\int\limits_D f(\pmb r-\pmb \xi)\nabla_{\pmb\xi}\cdot(\pmb Q \delta(\pmb \xi))dv_{\pmb\xi} = 
			\pmb Q\cdot\nabla_{\pmb r} \,f(\pmb r)
\end{equation}
where $f(\pmb r-\pmb \xi)$ is an arbitrary function; $\pmb \xi$  is integration variable; $\pmb r$ is radial distance.

Solution for the gradient vector stress function $\pmb B^*_g$ \eqref{pnf3} can be found in several steps. At first, we neglect the requirement that $\pmb B^*_g$ is the divergence-free field and find the particular solution $\tilde{\pmb B}^*_g$ to the equation:
\begin{equation}
\label{bgsol0}
	(1-l_2^2\nabla^2)\tilde{\pmb B}^*_g = -l_2^2\frac{\pmb Q}{\mu}\delta(\pmb r)
	\quad \implies\quad
	\tilde{\pmb B}^*_g = -\frac{\pmb Q}{4\pi\mu} \frac{e^{-r/l_2}}{r}
\end{equation}

Then we should subtract the potential part of the obtained solution $\tilde{\pmb B}^*_g$ to find the field $\pmb B^*_g$, which will satisfy $\nabla\cdot\pmb B^*_g=0$. Potential part of $\tilde{\pmb B}^*_g$ can be found by using Helmholtz decomposition and representation similar to \eqref{psi10}. As the result, we obtain:

\begin{equation}
\label{bgsol1}
	\pmb B^*_g = \tilde{\pmb B}^*_g - \nabla \Psi^*
\end{equation}
in which the potential $\Psi^*$ is defined according to relation (see \eqref{psi10})
\begin{equation}
\label{psibg}
\begin{aligned}
	\nabla\cdot\tilde{\pmb B}^*_g = \nabla^2 \Psi^*
\end{aligned}
\end{equation}

Then, combining equation for $\tilde{\pmb B}^*_g$ \eqref{bgsol0} and relation \eqref{psibg} and using \eqref{psh} we obtain:
\begin{equation}
\label{psibg2}
\begin{aligned}
	&\nabla\cdot\left(l_2^2\nabla^2\tilde{\pmb B}^*_g 
		-l_2^2\frac{\pmb Q}{\mu}\delta(\pmb r)\right) = \nabla^2 \Psi^*\\
	&\implies 
	\Psi^*
	= \frac{l_2^2}{4\pi\mu}\pmb Q\cdot\nabla 
			\left(\frac{1-e^{-r/l_2}}{r}\right) 
\end{aligned}
\end{equation}

Substituting \eqref{bgsol0}, \eqref{psibg2} into \eqref{bgsol1} we find the desirable solenoidal field $\pmb B^*_g$ in the following form:
\begin{equation}
\label{bgsol}
\begin{aligned}
	\pmb B^*_g = -\frac{\pmb Q}{4\pi\mu} \frac{e^{-r/l_2}}{r}
	 - \frac{l_2^2}{4\pi\mu} \pmb Q\cdot 
	 	\nabla \nabla\left(\frac{1 - e^{-r/l_2}}{r}\right)
\end{aligned}
\end{equation}

Solution for the gradient scalar stress function $\varphi^*_g$ \eqref{pnf4} is given by:
\begin{equation}
\label{phigsol}
\begin{aligned}
	&(1-l_1^2\nabla^2)\varphi^*_g = - l_1^2\,\frac{1-2\nu}{2(1-\nu)}\, 
					\frac{\nabla\cdot(\pmb Q \delta(\pmb r))}{\mu}\\
	&\implies\quad
	\varphi^*_g = - \frac{1-2\nu}{8\pi(1-\nu)\mu}\, 
					\pmb Q\cdot\nabla\left(\frac{e^{-r/l_1}}{r}\right)
\end{aligned}
\end{equation}

Substitution of relations \eqref{bcsol}, \eqref{phicsol}, \eqref{bgsol}, \eqref{phigsol} into the representation for the displacement field \eqref{pnf} provides us
\begin{equation}
\label{pnf22}
\begin{aligned}
	\pmb u &= \frac{\pmb Q}{4\pi\mu r}
	- \frac{1}{16\pi\mu(1-\nu)}\nabla
	\left(\frac{\pmb r\cdot \pmb Q}{r}\right) \\
	&+ l_1^2\frac{1-2\nu}{8\pi(1-\nu)\mu}\,\pmb Q\cdot\nabla\nabla
					\left(\frac{1}{r}\right)\\
	&- \frac{\pmb Q}{4\pi\mu} \frac{e^{-r/l_2}}{r}
	 - \frac{l_2^2}{4\pi\mu} \pmb Q\cdot 
	 	\nabla \nabla\left(\frac{1 - e^{-r/l_2}}{r}\right)\\
	 &- l_1^2 \frac{1-2\nu}{8\pi(1-\nu)\mu}\, 
					\pmb Q\cdot\nabla\nabla\left(\frac{e^{-r/l_1}}{r}\right)
\end{aligned}
\end{equation}

Finally, after some standard simplifications from \eqref{pnf22} we obtain the form of fundamental solution of SGET that exactly coincides to those one presented previously, e.g. in \cite{ma2018inclusion}:
\begin{equation}
\label{pnff}
\begin{aligned}
	\pmb u &= 
	\frac{1}{4\pi\mu}\left(
		\frac{\pmb Q}{r}
		- \frac{1}{2(1-\nu)}\pmb Q\cdot \nabla\nabla r
	\right)\\
	&+ l_1^2\frac{1-2\nu}{8\pi(1-\nu)\mu}\,\pmb Q\cdot\nabla\nabla
					\left(\frac{1-e^{-r/l_1}}{r}\right)\\
	&- \frac{\pmb Q}{4\pi\mu} \frac{e^{-r/l_2}}{r}
	 - \frac{l_2^2}{4\pi\mu} \pmb Q\cdot 
	 	\nabla \nabla\left(\frac{1 - e^{-r/l_2}}{r}\right)
\end{aligned}
\end{equation}

\section{Conclusion}
\label{con}
We derived and extended form of Boussinesq-Galerkin and Papkovich-Neuber general solution within Mindlin-Toupin strain gradient elasticity. We prove the theorem of completeness for these solutions, i.e. we show their generality \cite{truesdell1959invariant}. In the presented proof it is essential to have both forms of general solution (BG and PN) and established relations between different stress functions.  The further work should be related to the analysis of the nonuniqueness and the degree of nonuniqueness of the derived general solutions within SGET.

\section*{References}
\renewcommand{\bibsection}{}
\bibliography{refs.bib}

\begin{thebibliography}{10}

\bibitem{boussinesq1885application}
Joseph Boussinesq.
\newblock {\em Application des potentiels {\`a} l'{\'e}tude de l'{\'e}quilibre et du mouvement des solides {\'e}lastiques.}, volume~4.
\newblock Gauthier-Villars, 1885.

\bibitem{galerkin1930investigation}
B~Galerkin and BG~Galerkin.
\newblock On an investigation of stresses and deformations in elastic isotropic solids.
\newblock In {\em Dokl. Akad. Nauk SSSR, Ser. A}, pages 353--358, 1930.

\bibitem{papkovich1932representation}
PF~Papkovich.
\newblock The representation of the general integral of the fundamental equations of elasticity theory in terms of harmonic functions.
\newblock {\em Izv. Akad. Nauk SSSR, Phys.-Math. Ser}, 10(1425):90, 1932.

\bibitem{neuber1934neuer}
H~v Neuber.
\newblock Ein neuer ansatz zur l{\"o}sung r{\"a}umlicher probleme der elastizit{\"a}tstheorie. der hohlkegel unter einzellast als beispiel.
\newblock {\em ZAMM-Journal of Applied Mathematics and Mechanics/Zeitschrift f{\"u}r Angewandte Mathematik und Mechanik}, 14(4):203--212, 1934.

\bibitem{mindlin1936note}
RD156330362 Mindlin.
\newblock Note on the galerkin and papkovitch stress functions.
\newblock {\em Bulletin of the American Mathematical Society}, 42(6):373--376, 1936.

\bibitem{gurtin1962helmholtz}
ME~Gurtin.
\newblock On helmholtz's theorem and the completeness of the papkovich-neuber stress functions for infinite domains.
\newblock {\em Archive for Rational Mechanics and Analysis}, 9(1):225--233, 1962.

\bibitem{noll1957verschiebungsfunktionen}
W~Noll.
\newblock Verschiebungsfunktionen f{\"u}r elastische schwingungsprobleme.
\newblock {\em ZAMM-Journal of Applied Mathematics and Mechanics/Zeitschrift f{\"u}r Angewandte Mathematik und Mechanik}, 37(3-4):81--87, 1957.

\bibitem{sternberg1955concept}
Eli Sternberg and RA~Eubanks.
\newblock On the concept of concentrated loads and an extension of the uniqueness theorem in the linear theory of elasticity.
\newblock {\em Journal of Rational Mechanics and Analysis}, 4:135--168, 1955.

\bibitem{eubanks1956completeness}
RA~Eubanks and E~Sternberg.
\newblock On the completeness of the boussinesq-papkovich stress functions.
\newblock {\em Journal of Rational Mechanics and Analysis}, 5(5):735--746, 1956.

\bibitem{sokolnikoff1956mathematical}
Ivan~Stephen Sokolnikoff.
\newblock {\em Mathematical theory of elasticity}.
\newblock McGraw-Hill New York, 1956.

\bibitem{Slobodyansky1959}
M.~G. Slobodyansky.
\newblock {On the General and Complete Form of Solutions of the Equations of Elasticity}.
\newblock {\em Prikl. Mat. Mekh. (J. Appl. Math. Mech.)}, 23:468–482, 1959.

\bibitem{wang1991transformation}
Lin~Sheng Wang and Bin~Bing Wang.
\newblock The transformation of the papkovich-neuber (pn) general solution and others.
\newblock {\em Acta Mechanica Sinica}, 7(6):755--758, 1991.

\bibitem{lur1937theory}
AI~Lur’e.
\newblock On the theory of the system of linear differential equations with the constant coefficients.
\newblock {\em Trudy Leningrad, Industrial, in-ta}, 6:31--36, 1937.

\bibitem{wang1995completeness}
MZ~Wang and W13120260865 Wang.
\newblock Completeness and nonuniqueness of general solutions of transversely isotropic elasticity.
\newblock {\em International Journal of Solids and Structures}, 32(3-4):501--513, 1995.

\bibitem{wang2008recent}
MZ~Wang, BX~Xu, and CF~Gao.
\newblock Recent general solutions in linear elasticity and their applications.
\newblock {\em Applied Mechanics Reviews}, 61(3), 2008.

\bibitem{Mindlin1964}
R.~D. Mindlin.
\newblock {Micro-structure in linear elasticity}.
\newblock {\em Archive for Rational Mechanics and Analysis}, 16(1):51--78, 1964.

\bibitem{toupin1964theories}
R~Toupin.
\newblock Theories of elasticity with couple-stress.
\newblock {\em Archive for Rational Mechanics and Analysis}, 17(2):85--112, 1964.

\bibitem{dell2009generalized}
Francesco Dell'Isola, Giulio Sciarra, and Stefano Vidoli.
\newblock Generalized hooke's law for isotropic second gradient materials.
\newblock {\em Proceedings of the Royal Society A: Mathematical, Physical and Engineering Sciences}, 465(2107):2177--2196, 2009.

\bibitem{dell2016cauchy}
F~Dell'Isola, A~Madeo, and Pierre Seppecher.
\newblock Cauchy tetrahedron argument applied to higher contact interactions.
\newblock {\em Archive for Rational Mechanics and Analysis}, 219(3):1305--1341, 2016.

\bibitem{askes2011gradient}
Harm Askes and Elias~C Aifantis.
\newblock Gradient elasticity in statics and dynamics: an overview of formulations, length scale identification procedures, finite element implementations and new results.
\newblock {\em International Journal of Solids and Structures}, 48(13):1962--1990, 2011.

\bibitem{polizzotto2017hierarchy}
Castrenze Polizzotto.
\newblock A hierarchy of simplified constitutive models within isotropic strain gradient elasticity.
\newblock {\em European Journal of Mechanics-A/Solids}, 61:92--109, 2017.

\bibitem{gusev2017symmetry}
Andrei~A Gusev and Sergey~A Lurie.
\newblock Symmetry conditions in strain gradient elasticity.
\newblock {\em Mathematics and Mechanics of Solids}, 22(4):683--691, 2017.

\bibitem{mindlin1962effects}
RD~Mindlin and HF~Tiersten.
\newblock Effects of couple-stresses in linear elasticity.
\newblock {\em Archive for Rational Mechanics and Analysis}, 11(1):415--448, 1962.

\bibitem{lurie2021dilatation}
Sergey~A Lurie, Alexander~L Kalamkarov, Yury~O Solyaev, and Alexander~V Volkov.
\newblock Dilatation gradient elasticity theory.
\newblock {\em European Journal of Mechanics-A/Solids}, 88:104258, 2021.

\bibitem{eremeyev2020well}
Victor~A Eremeyev, Sergey~A Lurie, Yury~O Solyaev, et~al.
\newblock On the well posedness of static boundary value problem within the linear dilatational strain gradient elasticity.
\newblock {\em Zeitschrift f{\"u}r angewandte Mathematik und Physik}, 71(6):1--16, 2020.

\bibitem{eremeyev2018linear}
Victor~A Eremeyev, Claude Boutin, David Steigmann, et~al.
\newblock Linear pantographic sheets: existence and uniqueness of weak solutions.
\newblock {\em Journal of Elasticity}, 132(2):175--196, 2018.

\bibitem{gourgiotis2009plane}
PA~Gourgiotis and HG25720281193 Georgiadis.
\newblock Plane-strain crack problems in microstructured solids governed by dipolar gradient elasticity.
\newblock {\em Journal of the Mechanics and Physics of Solids}, 57(11):1898--1920, 2009.

\bibitem{askes2015understanding}
Harm Askes and Luca Susmel.
\newblock Understanding cracked materials: is linear elastic fracture mechanics obsolete?
\newblock {\em Fatigue \& Fracture of Engineering Materials \& Structures}, 38(2):154--160, 2015.

\bibitem{vasiliev2021new}
Valeriy Vasiliev, Sergey Lurie, and Yury Solyaev.
\newblock New approach to failure of pre-cracked brittle materials based on regularized solutions of strain gradient elasticity.
\newblock {\em Engineering Fracture Mechanics}, 258:108080, 2021.

\bibitem{lazar2013fundamentals}
Markus Lazar.
\newblock The fundamentals of non-singular dislocations in the theory of gradient elasticity: Dislocation loops and straight dislocations.
\newblock {\em International Journal of Solids and Structures}, 50(2):352--362, 2013.

\bibitem{makvandi2021strain}
Resam Makvandi, Bilen~Emek Abali, Sascha Eisentr{\"a}ger, and Daniel Juhre.
\newblock A strain gradient enhanced model for the phase-field approach to fracture.
\newblock {\em PAMM}, 21(1):e202100195, 2021.

\bibitem{cordero2016second}
Nicolas~M Cordero, Samuel Forest, and Esteban~P Busso.
\newblock Second strain gradient elasticity of nano-objects.
\newblock {\em Journal of the Mechanics and Physics of Solids}, 97:92--124, 2016.

\bibitem{rosi2018validity}
Giuseppe Rosi, Luca Placidi, and Nicolas Auffray.
\newblock On the validity range of strain-gradient elasticity: a mixed static-dynamic identification procedure.
\newblock {\em European Journal of Mechanics-A/Solids}, 69:179--191, 2018.

\bibitem{ma2014new}
HM~Ma and X-L Gao.
\newblock A new homogenization method based on a simplified strain gradient elasticity theory.
\newblock {\em Acta Mechanica}, 225(4):1075--1091, 2014.

\bibitem{lurie2011eshelby}
Sergey Lurie, Dmitrii Volkov-Bogorodsky, Anatolii Leontiev, and Elias Aifantis.
\newblock Eshelby’s inclusion problem in the gradient theory of elasticity: applications to composite materials.
\newblock {\em International Journal of Engineering Science}, 49(12):1517--1525, 2011.

\bibitem{seppecher2019pantographic}
Pierre Seppecher, Jean~Jacques Alibert, Tomasz Lekszycki, Roman Grygoruk, Marek Pawlikowski, David Steigmann, Ivan Giorgio, Ugo Andreaus, Emilio Turco, Maciej Go{\l}aszewski, et~al.
\newblock Pantographic metamaterials: an example of mathematically driven design and of its technological challenges.
\newblock {\em Continuum Mechanics and Thermodynamics}, 31(4):851--884, 2019.

\bibitem{dell2020discrete}
Francesco Dell'Isola and David~J Steigmann.
\newblock {\em Discrete and continuum models for complex metamaterials}.
\newblock Cambridge University Press, 2020.

\bibitem{solyaev2022self}
Yury Solyaev.
\newblock Self-consistent assessments for the effective properties of two-phase composites within strain gradient elasticity.
\newblock {\em Mechanics of Materials}, 169:104321, 2022.

\bibitem{lurie2006interphase}
S~Lurie, P~Belov, D~Volkov-Bogorodsky, and N~Tuchkova.
\newblock Interphase layer theory and application in the mechanics of composite materials.
\newblock {\em Journal of materials science}, 41(20):6693--6707, 2006.

\bibitem{charalambopoulos2015plane}
Antonios Charalambopoulos and Demosthenes Polyzos.
\newblock Plane strain gradient elastic rectangle in tension.
\newblock {\em Archive of Applied Mechanics}, 85(9):1421--1438, 2015.

\bibitem{charalambopoulos2022representing}
Antonios Charalambopoulos, Theodore Gortsas, and Demosthenes Polyzos.
\newblock On representing strain gradient elastic solutions of boundary value problems by encompassing the classical elastic solution.
\newblock {\em Mathematics}, 10(7):1152, 2022.

\bibitem{solyaev2019three}
Yury Solyaev, Sergey Lurie, and Vladimir Korolenko.
\newblock Three-phase model of particulate composites in second gradient elasticity.
\newblock {\em European Journal of Mechanics-A/Solids}, 78:103853, 2019.

\bibitem{placidi2017semi}
Luca Placidi and Amr~Ramadan El~Dhaba.
\newblock Semi-inverse method {\`a} la saint-venant for two-dimensional linear isotropic homogeneous second-gradient elasticity.
\newblock {\em Mathematics and Mechanics of Solids}, 22(5):919--937, 2017.

\bibitem{solyaev2022elastic}
Yury Solyaev, Sergey Lurie, Holm Altenbach, and Francesco dell’Isola.
\newblock On the elastic wedge problem within simplified and incomplete strain gradient elasticity theories.
\newblock {\em International Journal of Solids and Structures}, 239:111433, 2022.

\bibitem{stern1962}
E.~STERNBERG and M.~GuRTIN.
\newblock On the completeness of certain stress functions in the linear theory of elasticity.
\newblock {\em Proc. Fourth U.S. Nat. Cong. Appl. Mech.}, (44, 67):793--797, 1962.

\bibitem{gurtin1972}
M.~E. Gurtin.
\newblock {\em The Linear Theory of Elasticity (Encyclopedia of Physics Vol. 6 a/2)}.
\newblock Springer-Verlag, New York, 1972.

\bibitem{truesdell1959invariant}
C~Truesdell.
\newblock Invariant and complete stress functions for general continua.
\newblock {\em Archive for Rational Mechanics and Analysis}, 4(1):1--29, 1959.

\bibitem{charalambopoulos2020plane}
Antonios Charalambopoulos, Stephanos~V Tsinopoulos, and Demosthenes Polyzos.
\newblock Plane strain gradient elastic rectangle in bending.
\newblock {\em Archive of Applied Mechanics}, 90(5):967--986, 2020.

\bibitem{solyaev2021trefftz}
Yury~O Solyaev and Sergey~A Lurie.
\newblock Trefftz collocation method for two-dimensional strain gradient elasticity.
\newblock {\em International Journal for Numerical Methods in Engineering}, 122(3):823--839, 2021.

\bibitem{munch2017modified}
Ingo M{\"u}nch, Patrizio Neff, Angela Madeo, and Ionel-Dumitrel Ghiba.
\newblock The modified indeterminate couple stress model: Why yang et al.'s arguments motivating a symmetric couple stress tensor contain a gap and why the couple stress tensor may be chosen symmetric nevertheless.
\newblock {\em ZAMM-Journal of Applied Mathematics and Mechanics/Zeitschrift f{\"u}r Angewandte Mathematik und Mechanik}, 97(12):1524--1554, 2017.

\bibitem{jeffreys1956methods}
Harold Jeffreys and Bertha Swirles.
\newblock {\em METHODS OF MATHEMATICAL PHYSICS, by}.
\newblock Cambridge University Press USA, 1956.

\bibitem{petrascheck2015helmholtz}
D~Petrascheck and R~Folk.
\newblock The helmholtz decomposition of decreasing and weakly increasing vector fields.
\newblock {\em arXiv preprint arXiv:1506.00235}, 2015.

\bibitem{petrascheck2017helmholtz}
D~Petrascheck and R~Folk.
\newblock Helmholtz decomposition theorem and blumenthal's extension by regularization.
\newblock {\em arXiv preprint arXiv:1704.02287}, 2017.

\bibitem{lazar2019non}
Markus Lazar.
\newblock A non-singular continuum theory of point defects using gradient elasticity of bi-helmholtz type.
\newblock {\em Philosophical Magazine}, 99(13):1563--1601, 2019.

\bibitem{lazar2006dislocations}
Markus Lazar, G{\'e}rard~A Maugin, and Elias~C Aifantis.
\newblock Dislocations in second strain gradient elasticity.
\newblock {\em International Journal of Solids and Structures}, 43(6):1787--1817, 2006.

\bibitem{morse1954methods}
Philip~M Morse and Herman Feshbach.
\newblock Methods of theoretical physics.
\newblock {\em American Journal of Physics}, 22(6):410--413, 1954.

\bibitem{gao2009green}
X-L Gao and HM~Ma.
\newblock Green’s function and eshelby’s tensor based on a simplified strain gradient elasticity theory.
\newblock {\em Acta mechanica}, 207(3):163--181, 2009.

\bibitem{gourgiotis2018concentrated}
PA~Gourgiotis, Th~Zisis, and HG~Georgiadis.
\newblock On concentrated surface loads and green's functions in the toupin--mindlin theory of strain-gradient elasticity.
\newblock {\em International Journal of Solids and Structures}, 130:153--171, 2018.

\bibitem{ma2018inclusion}
Hansong Ma, Gengkai Hu, Yueguang Wei, and Lihong Liang.
\newblock Inclusion problem in second gradient elasticity.
\newblock {\em International Journal of Engineering Science}, 132:60--78, 2018.

\end{thebibliography}

\end{document}